\documentclass[a4paper,UKenglish,cleveref, autoref, thm-restate]{lipics-v2021}
\usepackage{tikz}
\usepackage{cleveref}
\usepackage{float}
\usepackage{graphicx}
\usepackage{float}
\usepackage{color}
\usepackage{braket}
\usepackage{xcolor}
\usepackage{tablefootnote}
\usepackage{graphicx}
\usepackage{amsmath}
\usepackage{amssymb}
\usepackage{amsthm}
\usepackage{amsfonts}
\usepackage{tcolorbox}
 \tcbuselibrary{skins , breakable , listings ,theorems}
\usepackage{colortbl}
\usepackage{algorithm}
\usepackage{algorithmic}

\title{Sample complexity of hidden subgroup problem} 

\titlerunning{Sample complexity of hidden subgroup problem} 

\author{Zekun Ye}{Institute of Quantum Computing and Computer Science Theory, School of Computer Science and Engineering, Sun Yat-sen University, Guangzhou 510006, China\\yezekun@mail2.sysu.edu.cn}{}{}{}

\author{Lvzhou Li}{Institute of Quantum Computing and Computer Science Theory, School of Computer Science and Engineering, Sun Yat-sen University, Guangzhou 510006, China \and Ministry of Education Key Laboratory of Machine Intelligence and Advanced Computing (Sun Yat-sen University), Guangzhou 510006, China \\lilvzh@mail.sysu.edu.cn}{lilvzh@mail.sysu.edu.cn}{}{}

\authorrunning{Z. Ye and L. Li} 

\Copyright{Zekun Ye and Lvzhou Li} 



\category{Track A: Algorithms, Complexity and Games} 

\relatedversion{} 

\nolinenumbers 

\begin{document}

\maketitle

\begin{abstract}
The hidden subgroup problem ($\mathsf{HSP}$) has been attracting much attention in quantum computing, since several well-known quantum algorithms including Shor algorithm can be described in a uniform framework as quantum methods to address different instances of it. One of the central issues about $\mathsf{HSP}$ is to characterize its quantum/classical  complexity.  For example, from the viewpoint of learning theory, sample complexity is a crucial concept.  However, 
while the quantum sample complexity of the problem has been studied,  a full characterization of the  classical sample complexity of $\mathsf{HSP}$ seems to be absent, which  will thus be the topic in this paper.  $\mathsf{HSP}$ over a  finite group is defined as follows: For a finite group $G$ and a finite set $V$, given a function $f:G \to V$ and the promise that for any $x, y \in G, f(x) = f(xy)$ iff $y \in H$ for a subgroup $H \in \mathcal{H}$, where $\mathcal{H}$ is a set of candidate subgroups of $G$, the goal is to identify $H$. Our contributions are as follows:
\begin{enumerate}[i)]
  \item For $\mathsf{HSP}$, we show that the number of    uniform  examples necessary to learn the hidden subgroup with bounded error is at least $\Omega\left(\max \left\{\min\limits_{H \in\mathcal{H}}\frac{\log |\mathcal{H}|}{\log \frac{|G|}{|H|}}, \min\limits_{H \in\mathcal{H}}\sqrt{\frac{|G|}{|H|}\frac{\log |\mathcal{H}|}{\log \frac{|G|}{|H|}}}\right\}\right)$, and on the other hand, $O\left(\max\limits_{H \in \mathcal{H}} \left\{sr(\mathcal{H}),\sqrt{\frac{|G|}{|H|}sr(\mathcal{H})}\right\}\right)$   uniform examples are sufficient, where 
$sr(\mathcal{H}) = \max\limits_{H \in \mathcal{H}}\mathsf{r}(H)$ and $\mathsf{r}(H)$ is the rank of $H$.
  \item  By concretizing the parameters of $\mathsf{HSP}$, we consider a class of restricted Abelian hidden subgroup problem ($\mathsf{rAHSP}$) and obtain the upper and lower bounds for the sample complexity of $\mathsf{rAHSP}$. 
  \item  We continue to discuss a special case of $\mathsf{rAHSP}$, generalized Simon's problem ($\mathsf{GSP}$), and show that the sample complexity of $\mathsf{GSP}$ is $\Theta\left(\max\left\{k,\sqrt{k\cdot p^{n-k}}\right\}\right)$. Thus we obtain a complete characterization of the sample complexity of $\mathsf{GSP}$. 
\end{enumerate}
\end{abstract}
\section{Introduction}
\subsection{Background}
\textbf{Hidden subgroup problem}. The hidden subgroup problem plays an important role in the history of quantum computing. 
Several important quantum algorithms such as Deutsch-Jozsa algorithm \cite{deutsch1992rapid}, Simon algorithm \cite{simon1994power}, and Shor algorithm \cite{shor1994algorithms} have a uniform description in the framework of the hidden subgroup problem \cite{jozsa1998quantum}. Moreover, many quantum algorithms were proposed for the instances of the hidden subgroup problem, e.g., \cite{bacon2005from,childs2007quantum,ettinger2004the,grigni2001quantum,hallgren2003the,kempe2005the,kuperberg2005a}.

The hidden subgroup problem consists of the Abelian hidden subgroup problem and the non-Abelian hidden subgroup problem. Many problems are special cases of the Abelian hidden subgroup problem, such as Simon's problem \cite{simon1994power}, generalized Simon's problem \cite{Ye2019query} and some important number-theoretic problems \cite{Hallgren2007Polynomial,Hallgren2005Fast,Schmidt2005Polynomial}. The non-Abelian hidden subgroup problem also received much attention \cite{grigni2001quantum,IvanyosMS01,kuperberg2005a,FriedlIMSS03,hallgren2003the,MooreRRS07,Regev04}. While there exist efficient quantum algorithms to solve the Abelian hidden subgroup problem \cite{simon1994power,BonehL95,Kitaev1995quantum,brassard1997exact,MoscaE98,EisentragerHK014}, many instances of the non-Abelian hidden subgroup problem are not known to have efficient quantum algorithms, such as the dihedral hidden subgroup problem and the symmetric hidden subgroup problem \cite{kuperberg2005a, Decker2013Hidden}. 

There exist two versions of the hidden subgroup problem: the identification and decision versions. The task of the identification version is to identify the hidden subgroup, whereas the task of the decision version is to decide whether the hidden subgroup is the trivial group or not. In this paper, all the hidden subgroup problems mentioned belong to the identification version without special instructions. 
\\
\\
\noindent
\textbf{Sample and query complexity}. In learning theory, there are two types of learning models: passive learning and active learning \cite{Hanneke14,BalcanHV10}. In passive learning, the algorithm can only receive random labeled examples in a passive way; in active learning, the algorithm can interactively ask for the labels of examples of its own choosing. Thus, active learning may enable us to design more powerful algorithms compared to passive learning for the same problem.

Specifically, we will focus on the task to learn the property about some function $f$ with high success probability in the following. In passive learning, an algorithm can obtain i.i.d. random labeled examples $(x, f(x))$, where $x$ is distributed according to a certain probability distribution. Such an algorithm is called a sample algorithm. The sample complexity of a sample algorithm is the maximum number of i.i.d. random examples needed for learning the property about $f$ in the worst case. On the other hand, in active learning, an algorithm is allowed to make queries. The algorithm can choose some $x$ to learn $f(x)$ in each query. Such an algorithm is called a query algorithm. The query complexity of a query algorithm is the maximum number of queries needed for learning the property about $f$ in the worst case. The \textit{sample} (\textit{query}) \textit{complexity} of a learning problem is the sample (query) complexity of the optimal sample (query) algorithm. 

A query algorithm may make one query at a time, using the information from previous queries to decide which example to query next. Thus, a query algorithm may cost less than a sample algorithm. That is, an upper bound on the sample complexity is always an upper bound on the query complexity for the same problem, but a lower bound on the sample complexity is not necessarily a lower bound on the query complexity. 

Much work analyzed the quantum advantage via the sample complexity \cite{ArunachalamW17}. For example, the concept class of DNF-formulas and ($\log n$)-juntas can be learned in polynomial time from quantum examples under the uniform distribution \cite{BshoutyJ99, AticiS07}, whereas the best known classical algorithm of these two problems both runs in quasi-polynomial time under the uniform examples \cite{Verbeurgt90, MosselOS04}. Arunachalam and de Wolf \cite{Arunachalam2017optimal} proved that quantum and classical sample complexity are equal up to constant factors in both the PAC and agnostic models.
Arunachalam et al. \cite{Arunachalam2019two} showed a $k$-Fourier-sparse $n$-bit Boolean function can be learned from $O(k^{1.5}(\log k)^2)$ uniform quantum examples for that function, whereas $\Omega(nk)$ uniform examples are necessary in the classical case \cite{HavivR16}. 
\\
\\

\noindent
\textbf{Sample complexity of the hidden subgroup problem}. The quantum sample complexity of the hidden subgroup problem has been attracting great attention. Bacon et al. \cite{bacon2005from} presented the quantum sample complexity of the hidden subgroup problem over semidirect product group $A \rtimes \mathbb{Z}_p$ is $\Theta(\frac{\log |A|}{\log p})$, where $A$ is any Abelian group and $p$ is a prime. Ettinger et al. \cite{ettinger2004the} proved that the quantum sample complexity of the hidden subgroup problem over any finite group $G$ is $O(\log^2 |G|)$. In terms of the candidate set $\mathcal{H}$ of hidden subgroups, Moore and Russell \cite{MooreR07} showed the quantum sample complexity of a wide class of the hidden subgroup problem is $O(\log |\mathcal{H}|)$. Furthermore, Hayashi et al. \cite{HayashiKK08} proved that the quantum sample complexity of the hidden subgroup problem is at most $O\left(\frac{\log |\mathcal{H}|}{\log \min_{H \neq H' \in \mathcal{H}}(|H|/|H \cap H'|)}\right)$ and at least $\Omega\left(\frac{\log |\mathcal{H}|}{\log \max_{H \in \mathcal{H}}|H|}\right)$; if all the candidate subgroups in $\mathcal{H}$ have the same prime order $p$, then the quantum sample complexity is $\Theta(\frac{\log \mathcal{H}}{\log p})$. 
 
On the other hand, to our knowledge, almost no related direct result has been obtained in terms of the classical sample complexity of the hidden subgroup problem. However, there exists only some discussion about the classical query complexity for some instances of the hidden subgroup problem. For example, the classical query complexity of Simon's problem was proven to be $\Theta(\sqrt{2^n})$ \cite{simon1994power,cai2018optimal,Wolf2019quantum}. Ye et al. \cite{Ye2019query} proved that a nearly optimal bound for the classical query complexity of generalized Simon's problem ($\mathsf{GSP}$). For the order-finding problem over $\mathbb{Z}_{2^m} \times \mathbb{Z}_{2^n}$, Cleve \cite{Cleve04} proved that the deterministic query complexity is at least $\Omega\left(\sqrt{\frac{2^n}{m}}\right)$, and the bounded-error query complexity is at least $\Omega\left(\frac{2^{n/3}}{\sqrt{m}}\right)$. Kuperberg \cite{kuperberg2005a} proved the classical query complexity of the dihedral hidden subgroup problem over the dihedral group $D_n$ is $\Omega(\sqrt{N})$. Childs \cite{Childs2021Lecture} showed that a classical algorithm must make $\Omega(\sqrt{N})$ queries if there are $N$ candidate subgroups whose only common element is the identity element. Recently, Nayak \cite{Nayak2021} proposed the deterministic query algorithms for solving the hidden subgroup problem. 

It is worth noting that a lower bound on the classical query complexity is also a lower bound on the classical sample complexity,  but not necessarily a tight lower bound. For example, for GSP, its classical sample complexity will be shown to be $\Theta\left(\max\left\{k,\sqrt{k\cdot p^{n-k}}\right\}\right)$ in this paper, whereas the known best  lower bound on the classical query complexity was given as  $\Omega\left(\max\{k, \sqrt{p^{n-k}}\}\right)$  in \cite{Ye2019query} .
\\
\\
\noindent
\textbf{Motivation}.
The motivation for studying the classical sample complexity of the hidden subgroup problem is as follows: (i) Note that there exists a great understanding of the quantum sample complexity for the hidden subgroup problem. However, as far as we know, how well do classical sample algorithms perform on this problem still needs to be explored. (ii) For some instances of the hidden subgroup problem, such as $\mathsf{GSP}$, the classical query complexity is not tight \cite{Ye2019query}. Due to the difficulty in exploring query complexity, we hope to obtain a better lower bound on the classical sample complexity of  $\mathsf{GSP}$, since the sample model is weaker than the query model.
\subsection{Problem statement and our results}\label{sec:resut}
In this paper, we consider the classical sample complexity of the hidden subgroup problem ($\mathsf{HSP}$) over any finite group. The definition of $\mathsf{HSP}$ is as follows:
\begin{definition}[$\mathsf{HSP}$]
\
\begin{tcolorbox}
			\label{identification version}
			\noindent
			\textbf{Given:} A finite group $G$; a set $\mathcal{H}$ of candidate subgroups of $G$; an (unknown) function $f:G \to V$, where $V$ is a finite set.
			\\
			\\
			\textbf{Promise:} There exists a subgroup $H \in \mathcal{H}$ such that for any $x, y \in G, f(x) = f(xy)$ iff $y \in H$.
			\\
			\\
			\textbf{Problem:} Identify $H \in \mathcal{H}$.
		\end{tcolorbox}
\end{definition}
Unlike the general definition, we explicitly give the candidate subgroups set $\mathcal{H}$, which is a critical  component in the hidden subgroup problem. Actually, the hidden subgroup problem only depends on $G$ and $\mathcal{H}$ essentially. 

Moreover, we consider some interesting instances of $\mathsf{HSP}$ further by giving more concrete parameters. First, we discuss an instance of the hidden subgroup problem over a class of Abelian groups, call the restricted Abelian hidden subgroup problem ($\mathsf{rAHSP}$). The definition of $\mathsf{rAHSP}$ is as follows:
\begin{definition}[$\mathsf{rAHSP}$]\
\begin{tcolorbox}
			\noindent
			\textbf{Given:} An (unknown) function $f:G \to V$, where $G = \mathbb{Z}_{p_1}^{n_1} \times \mathbb{Z}_{p_2}^{n_2} \times \cdots \times \mathbb{Z}_{p_m}^{n_m}$ and $p_i$'s are primes; positive integers $k_1,...,k_m$ satisfying that $k_i < n_i$ for any $i \in [m]$.
			\\
			\\
			\textbf{Promise:} There exists a subgroup $H$ such that (i) $H = H_1 \times H_2 \times \cdots \times H_m$; (ii) $rank(H_i) = k_i$ for any $i \in [m]$; (iii) for any $x, y \in G, f(x) = f(x+y)$ iff $y \in H$.
			\\
			\\
			\textbf{Problem:} Identify $H$.
\end{tcolorbox}
\end{definition}
It is easy to see $\mathsf{rAHSP}$ is a subproblem of $\mathsf{HSP}$, since $\mathcal{H}$ is the set of all the subgroups satisfying the above promise. Note that any candidate subgroup in $\mathcal{H}$ has the same order in $\mathsf{rAHSP}$. 

Furthermore, we continue to consider a simplified version of $\mathsf{rAHSP}$, generalized Simon's problem  ($\mathsf{GSP}$),  defined as follows:
\begin{definition}[$\mathsf{GSP}$ \cite{Ye2019query}]\
\begin{tcolorbox}
			\label{generalized Simon's problem}
			\noindent
			\textbf{Given:} An (unknown) function $f:\mathbb{Z}_p^n \to V$, where $p$ is a prime, $V$ is a finite set; a positive integer $k <n$.
			\\
			\\
			\textbf{Promise:} There exists a subgroup $H \le \mathbb{Z}_p^n$ of rank $k$ such that for any $x, y \in \mathbb{Z}_p^n, f(x) = f(y)$ iff $y-x \in H$.
			\\
			\\
			\textbf{Problem:} Identify $H$.
		\end{tcolorbox}
\end{definition}
In $\mathsf{GSP}$, $G = \mathbb{Z}_p^n$, $\mathcal{H}$ is the set of all the subgroups of rank $k$ in $G$. Additionally, $\mathsf{GSP}$ is an extension version of Simon's problem, which is a well-known problem in the history of quantum computing. Simon's problem is a special case of $\mathsf{GSP}$ with $k = 1$ and $p = 2$, as shown in Definition \ref{Def:Simon}. Similarly, we can express Simon's problem as an instance of $\mathsf{HSP}$ by making $G = \mathbb{Z}_2^n$, $\mathcal{H} = \{\{0,s\}|s\in \{0,1\}^n/ \{0\}\}$.

\begin{definition}[Simon's problem]\
\label{Def:Simon}
\begin{tcolorbox}
			\noindent
			\textbf{Given:} An (unknown) function $f:\mathbb{Z}_2^n \to V$, where $V$ is a finite set.
 						\\
			\\
			\textbf{Promise:} There exists a non-zero string $s \in \{0,1\}^n$ such that for any $x, y \in \mathbb{Z}_2^n$, $f(x) = f(y)$ iff $x = y$ or $x + y = s$.
			\\
			\\
			\textbf{Problem:} Identify $s$.
\end{tcolorbox}
\end{definition}

The above problems have the following relationship:

\tikzstyle{class}=[]
\tikzstyle{arrow} = [->]
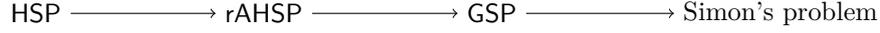
\begin{figure}[H]
\centering
\begin{tikzpicture}[node distance=0cm]
\node[class](rootnode){$\mathsf{HSP}$};
\node[class,below of=rootnode,yshift=0,xshift=3cm](first-1){$\mathsf{rAHSP}$};
\node[class,below of=first-1,yshift=0,xshift=3cm](two-1){$\mathsf{GSP}$};
\node[class,below of=two-1,yshift=0,xshift=3.8cm](three-1){Simon's problem};
\draw [arrow] (rootnode) -- node [font=\small] {} (first-1);
\draw [arrow] (first-1) -- node [font=\small] {} (two-1);
\draw [arrow] (two-1) -- node [font=\small] {} (three-1);
\end{tikzpicture}
\caption{The relations between the above problems. A rightward arrow from $A$ to $B$ means $B$ is a subproblem of $A$.}
\label{decision tree}
\end{figure}

For the above problems, given enough i.i.d. examples, a learning algorithm may give correct answers with bounded error $\delta$ ($0 \le \delta < 1/2$). In this case, we define the sample complexity of the problem as the number of examples needed by the optimal learning algorithm in the worst case. 
In the classical case, an example has the form $(x,f(x))$, where $x$ is distributed according to  a given  distribution over $G$. 
If $x$ is assumed to follow the uniform distribution, then $(x,f(x))$ is called a {\it uniform example}. In the quantum case,  a uniform quantum example is such a quantum state   $\frac{1}{\sqrt{|G|}}\sum_{x \in G} \ket{x}\ket{f(x)}$. In this paper, we focus on the classical case, and our problem is: what number of uniform examples is sufficient and necessary to learn the goal with bounded error in the above problems? 

We first obtain some characterizations for the sample complexity of $\mathsf{HSP}$. Let $sr(\mathcal{H}) = \max_{H \in \mathcal{H}}\min\{|S|:S\subseteq H, \langle S \rangle = H\}$. Our main result is as follows:
 \begin{theorem}
 \label{Theorem:main}
For $\mathsf{HSP}$, the number of uniform examples necessary to learn the hidden subgroup with bounded error is at least $\Omega\left(\max \left\{\min\limits_{H \in\mathcal{H}}\frac{\log |\mathcal{H}|}{\log \frac{|G|}{|H|}}, \min\limits_{H \in\mathcal{H}}\sqrt{\frac{|G|}{|H|}\frac{\log |\mathcal{H}|}{\log \frac{|G|}{|H|}}}\right\}\right)$. On the other hand, the number of uniform examples sufficient to learn the hidden subgroup with bounded error is at most $O\left(\max\limits_{H \in \mathcal{H}} \left\{sr(\mathcal{H}),\sqrt{\frac{|G|}{|H|}sr(\mathcal{H})}\right\}\right)$.
\end{theorem}

We also analyze the sample complexity of  $\mathsf{rAHSP}$ and $\mathsf{GSP}$. By Theorem \ref{Theorem:main}, we obtain the following corollaries:
\begin{corollary}
\label{corollary:ahsp}
For $\mathsf{rAHSP}$, the number of uniform examples necessary to learn the hidden subgroup with bounded error is at least  $\Omega\left(\max\left\{\min\limits_{i\in [m]} k_i, \min\limits_{i\in [m]} \sqrt{k_i \prod_{j=1}^m p_j^{n_j-k_j}}\right\}\right)$. Moreover, the number of uniform  examples sufficient to learn the hidden subgroup with bounded error is at most $O\left(\max\limits_{i\in [m]}\left\{k_i, \sqrt{k_i \prod_{j=1}^m p_j^{n_j-k_j}}\right\} \right)$.
\end{corollary}

\begin{corollary}
\label{corollary:gsp}
The sample complexity of $\mathsf{GSP}$ is $\Theta\left(\max\left\{k,\sqrt{k\cdot p^{n-k}}\right\}\right)$.
\end{corollary}

Additionally, the sample complexity of Simon problem is a well-known result as Claim \ref{Claim:sp} (e.g. \cite{simon1994power,Wolf2019quantum}). Corollary \ref{corollary:gsp} matches with Claim \ref{Claim:sp} when $k = 1$ and $p=2$.
\begin{claim}
\label{Claim:sp}
The sample complexity of Simon's problem is $\Theta(\sqrt{2^n})$.
\end{claim}

We list  lower bounds on the sample complexity of Simon's problem, $\mathsf{GSP}$, $\mathsf{HSP}$ in Table \ref{table:lower bound}. The methods in this paper may be helpful to other problems.

\begin{table}[H]
			\caption{Known results about  the sample complexity of Simon's problem, $\mathsf{GSP}$ and $\mathsf{HSP}$. The results of $\mathsf{GSP}$ and $\mathsf{HSP}$ are obtained in this paper.}
			\label{table:lower bound}
			\centering
			\scalebox{0.66}{
			\begin{tabular}{|c|c|c|c|c|}
				\hline
				\ & Simon's problem & $\mathsf{GSP}$ & $\mathsf{HSP}$ \\
				\hline
				classical & $\Theta(\sqrt{2^n})$\cite{simon1994power,Wolf2019quantum} & \textcolor{orange}{$\Theta\left(\max\{k,\sqrt{k\cdot p^{n-k}}\}\right)$}
				 &  \textcolor{orange}{$O\left(\max\limits_{H \in \mathcal{H}} \left\{sr(\mathcal{H}),\sqrt{\frac{|G|}{|H|}sr(\mathcal{H})}\right\}\right)$, $\Omega\left(\max \left\{\min\limits_{H \in\mathcal{H}}\frac{\log |\mathcal{H}|}{\log \frac{|G|}{|H|}}, \min\limits_{H \in\mathcal{H}}\sqrt{\frac{|G|}{|H|}\frac{\log |\mathcal{H}|}{\log \frac{|G|}{|H|}}}\right\}\right)$} \\
				\hline
				quantum & $\Theta(n)$\tablefootnote{First, Simon algorithm \cite{simon1994power} is a sample algorithm with $O(n)$ uniform examples. Second, Koiran et al. \cite{KoiranNP05} presented the lower bound on the quantum query complexity of Simon's problem is $\Omega(n)$, so this is also a lower bound on the quantum sample complexity of Simon's problem.} & $\Theta(n-k)$\tablefootnote{$\mathsf{GSP}$ can be solved with a generalized Simon's alogrithm with $O(n-k)$ uniform examples \cite{Hirvensalo2001quantum}. Additionally, by substiting $|\mathcal{H}| = \prod_{j=0}^{k-1} \frac{p^{n}-p^j}{p^{k}-p^j}$ and $|H| = p^k$ ($\forall H \in \mathcal{H}$) \cite{Ye2019query} into $\Omega\left(\frac{\log |\mathcal{H}|}{\log \max_{H \in \mathcal{H}}|H|}\right)$ \cite{HayashiKK08}, we can see the lower bound on the quantum sample complexity of $\mathsf{GSP}$ is $\Omega(n-k)$.} & $O\left(\frac{\log |\mathcal{H}|}{\log \min\limits_{H \neq H' \in \mathcal{H}}(|H|/|H \cap H'|)}\right)$, $\Omega\left(\frac{\log |\mathcal{H}|}{\log \max\limits_{H \in \mathcal{H}}|H|}\right)$\cite{HayashiKK08}\\
				\hline
			\end{tabular}
			}
		\end{table}

 \subsection{Organization}
	The remainder of the paper is organized as follows. In Section \ref{sec:Pre}, we review some notations used in this paper.
	In Section \ref{sec:Bounds}, we present the lower and upper bounds of sample complexity of $\mathsf{HSP}$. 
	In Section \ref{sec:Application}, we apply the result in Section \ref{sec:Bounds} to obtain the sample complexity of some special problems, including $\mathsf{rAHSP}$ and $\mathsf{GSP}$. Finally, a conclusion is made in Section \ref{sec:Conclusion}.

\section{Preliminary}
\label{sec:Pre}
In this section, we present some notations used in this paper. Let $[m] = \{1,2,...,m\}$
and $\mathbb{Z}_p$ denote the additive group of elements $\{0,1,...,p-1\}$ with addition modulo $p$ denoted by $+$. For two groups $G_1,G_2$, let $G_1 \times G_2$ denote the direct product of $G_1$ and $G_2$. For a finite group $G$, a subset $S$ is said to be a generating set for $G$ if all elements in $G$ can be expressed as the finite product of elements in $S$ and their inverses, i.e., $G = \langle S \rangle = \{a_1^{l_1}a_2^{l_2}\cdots a_k^{l_k}|a_i \in S, l_i = \pm 1, k\in \mathbb{N}\}$. 
The \textit{rank} of $G$ is the cardinality of a minimal generating set of $G$, denoted by $\mathsf{r}(G) = \min\{|S|:S\subseteq G, \langle S \rangle = G\}$.
If $H$ is a subgroup of $G$, then $H \le G$; if $H$ is a proper subgroup of $G$, then $H < G$.
 Note that if $H$ is a subgroup of $\mathbb{Z}_p^n$, then $\mathsf{r}(H) = k$ if and only if $|H| = p^k$.
For a set $\mathcal{H} $ consisting of subgroups of $G$, the \textit{subgroup rank} of $\mathcal{H}$ is defined as $\mathsf{sr}(\mathcal{H}) = \max_{H \in \mathcal{H}}\mathsf{r}(H)$. The group with only one element, the identity element, is called the \textit{trivial group}.
 
Suppose $X,Y,Z$ are discrete random variables. 
If $X \sim p(x)$, then the  \textit{Shannon entropy} associated with $X$ is defined as $I(X) = -\sum_{x}p(x) \log p(x)$. 
If $(X,Y) \sim p(x,y)$, then the \textit{joint entropy} of $X$ and $Y$ is defined as $I(X,Y) = -\sum_{x,y}p(x,y)\log p(x,y)$; the \textit{entropy} of $X$ conditional on knowing $Y$ is defined as $I(X|Y) = I(X,Y)-I(Y)$; the \textit{mutual information} of $X$ and $Y$ is defined as $I(X:Y) = I(X)+I(Y)-I(X,Y)$; the \textit{conditional mutual information} of $X$ and $Y$ conditional on knowing $Z$ is defined by $I(X:Y|Z) = I(X|Z) - I(X|Y,Z)$. The \textit{binary entropy} of a bit with distribution $(p,1-p)$ is defined as $I(p) = -p \log p-(1-p)\log (1-p)$. Some basic properties of Shannon entropy \cite{CT2001} are useful in this paper:

\begin{itemize}
\item $I(X_1,X_2,...,X_n) \le \sum_{i=1}^n I(X_i)$ with equality if and only if $X_i$'s are independent random variables.
\item $I(X:Y) \ge 0$ and $I(X|Y) \le I(X)$ with equality if and only if $X$ and $Y$ are independent. 
\item $I(X:Y) \le I(Y)$ with equality if and only if $Y$ is a function of $X$.
\item $I(X) \le \log |\mathcal{X}|$ with equality if and only if $X$ is a uniform random variables over $\mathcal{X}$.

\item If $I(X:Z|Y) = 0$, then $I(X:Y) \ge I(X:Z)$. 
\end{itemize}

For $\mathsf{HSP}$, let $\mathscr{H}$ be the random variable of the hidden subgroup that is uniformly distributed over $\mathcal{H}$. 
For a sample algorithm of $\mathsf{HSP}$, suppose the number of i.i.d. uniform examples is $T$. Let $B_i = (X_i,f(X_i))$ be the random variable of the $i$-th uniform examples for $i \in [T]$. Let $B = B_1\cdots B_T$, $X = X_{1}\cdots X_{T}$ and $f(X)=f(X_1)\cdots f(X_T)$. Suppose the range of $B$ is $\mathcal{B}$. For $i,j \in [T]$, we define random variable $Y_{X_i,X_j}$ as follows: if $f(X_i) = f(X_j)$, let $Y_{X_i,X_j} = 1$; if $f(X_i) \neq f(X_j)$, let $Y_{X_i,X_j} = 0$. Let $Y$ be the sequence of $Y_{X_i,X_j}$ for any $i < j$, i.e., $Y = Y_{X_1,X_2}Y_{X_1,X_3}\cdots Y_{X_{T-1},X_T}$. 
Since $Y$ is a function of $B$, without loss of generality, we assume $g(B) = Y$.

\section{General Bounds for $\mathsf{HSP}$}
\label{sec:Bounds}
In this section, we present the general bounds for the sample complexity of $\mathsf{HSP}$ in Theorem \ref{Theorem:main}. 
We show  lower and upper bounds in Section \ref{sec:lower} and \ref{sec:upper},  respectively.

\subsection{Lower bound}
\label{sec:lower}
In this section, we present a lower bound by proving Theorem \ref{Theorem:low}. We use an information-theoretic method.


\begin{theorem}[Lower bound]
\label{Theorem:low}
The number of uniform examples necessary to solve $\mathsf{HSP}$ with bounded error is at least $\Omega\left(\max\left\{\min\limits_{H \in \mathcal{H}}\frac{\log |\mathcal{H}|}{\log \frac{|G|}{|H|}},\min\limits_{H \in \mathcal{H}}\sqrt{\frac{|G|}{|H|}\frac{\log |\mathcal{H}|}{\log \frac{|G|}{|H|}}}\right\}\right)$.
\end{theorem}

\begin{proof}


First, we prove the left part of the lower bound using the following three-step analysis: 
\begin{enumerate}
\item $I(f(X)) \ge (1-\delta)\log |\mathcal{H}|-H(\delta)$. 

\textit{Proof of item 1}. 
Since
\begin{equation*}
\begin{aligned}
I(\mathscr{H}:B) &\le I(\mathscr{H}:X)+I(\mathscr{H}:f(X)) \\
&= 0+I(\mathscr{H}:f(X)) \\
&\le I(f(X)),\\
\end{aligned}
\end{equation*}
we have $I(f(X)) \ge I(\mathscr{H}:B) > (1-\delta)\log |\mathcal{H}|-I(\delta)$ by Lemma \ref{Lemma:HSB}.

\item $I(f(X)) \le \sum_{i} I(f(X_i))$.


\item $I(f(X_i)) \le \max\limits_{H\in \mathcal{H}}\log \frac{|G|}{|H|}$ for any $i$.

\textit{Proof of item 3.}
The function value of $f$ is constant on cosets of $H$ and distinct among different cosets of $H$, thus $f(X_i)$ is a uniform variable over $\frac{|G|}{|H|}$ different values. In the worst case, $H$ is the largest subgroup of $G$. Thus, $I(f(X_i)) \le \max\limits_{H\in \mathcal{H}}\log \frac{|G|}{|H|}$.  
\end{enumerate}
Combining these three steps implies 
\begin{equation*}
\begin{aligned}
T &> \frac{(1-\delta)\log |\mathcal{H}|-I(\delta)}{\max\limits_{H\in \mathcal{H}}\log \frac{|G|}{|H|}} \\
&= \min\limits_{H\in \mathcal{H}}\frac{(1-\delta)\log |\mathcal{H}|-I(\delta)}{\log \frac{|G|}{|H|}} \\
&= \Omega\left(\min\limits_{H\in \mathcal{H}} \frac{\log |\mathcal{H}|}{\log \frac{|G|}{|H|}}\right). \\
\end{aligned}
\end{equation*}

Second, we continue to prove the right part using a similar method.
\begin{enumerate}
\item $I(Y) \ge (1-\delta)k \cdot \log |\mathcal{H}|-H(\delta)$. 

\textit{Proof of item 1.}
By the definition of $B$ and $Y$, we have 
\begin{equation*}
Pr\{B=b|Y=y\} = 
\begin{cases}
\frac{1}{|\{b\in \mathcal{B}:g(b) = y\}|},& \text{if } g(b) = y \\
0,& \text{if } g(b) \neq y
\end{cases}.
\end{equation*}
Thus, given $Y$, $B$ is independent of $\mathscr{H}$, i.e, $Y$ is a sufficient statistic of $\mathscr{H}$, which means $I(\mathscr{H}:B|Y) = 0$. As a result, $I(\mathscr{H}:Y) \ge I(\mathscr{H}:B)$. Hence, 
$I(Y) \ge I(\mathscr{H}:Y) \ge I(\mathscr{H}:B) > (1-\delta)k \cdot \log |\mathcal{H}|-I(\delta)$ by Lemma \ref{Lemma:HSB}.

\item $I(Y) \le \sum_{i < j} I(Y_{X_i,X_j})$.

\item $I(Y_{X_i,X_j}) \le \max_{H \in \mathcal{H}}2\frac{|H|}{|G|} \log \frac{|G|}{|H|}$.

\textit{Proof of item 3.}
By definition of $\mathsf{HSP}$, $Y_{X_i,X_j} = 1$ iff $X_i^{-1}X_j \in H$. Since $X_i,X_j$ are independent uniform variables over $G$, $(X_i)^{-1}X_j$ is also a uniform variable over $G$. Hence, 
\begin{equation*}
    Pr\{Y_{X_i,X_j}=1\} = Pr\{(X_i)^{-1}X_j \in H\} = \frac{|H|}{|G|},
\end{equation*}
so
\begin{equation*}
I(Y_{X_i,X_j}) \le \max\limits_{H \in \mathcal{H}}I(\frac{|H|}{|G|}) \le \max\limits_{H \in \mathcal{H}}2\frac{|H|}{|G|} \log \frac{|G|}{|H|},
\end{equation*} 
where the last inequality follows by Claim \ref{Claim:entropy}. 

\end{enumerate}
Combining these three steps implies 
\begin{equation*}
\binom{T}{2} > \frac{(1-\delta) \cdot \log |\mathcal{H}|-H(\delta)}{\max\limits_{H \in \mathcal{H}}2\frac{|H|}{|G|} \log \frac{|G|}{|H|}},
\end{equation*}
which means 
\begin{equation*}
T =\Omega\left(\min\limits_{H \in \mathcal{H}}\sqrt{\frac{|G|}{|H|}\frac{\log |\mathcal{H}|}{\log \frac{|G|}{|H|}}}\right).
\end{equation*}
Finally, we have $\Omega\left(\max\left\{\min\limits_{H \in \mathcal{H}}\frac{\log |\mathcal{H}|}{\log \frac{|G|}{|H|}},\min\limits_{H \in \mathcal{H}}\sqrt{\frac{|G|}{|H|}\frac{\log |\mathcal{H}|}{\log \frac{|G|}{|H|}}}\right\}\right)$.

\end{proof}

\begin{lemma}
\label{Lemma:HSB}
$I(\mathscr{H}:B) > (1-\delta)\log |\mathcal{H}|-H(\delta)$.
\end{lemma}
\begin{proof}
 Let random variable $\mathscr{H}_B$ be the hypothesis that the learner produces (given $B$). According to the setting of learning algorithms, it is required that $Pr\{\mathscr{H} \neq \mathscr{H}_B\} \le \delta$. By Fano inequality \cite{CT2001}, we have $I(\mathscr{H}|B)\le I(\delta)+\delta \log (|\mathcal{H}|-1)$. Thus, 
\begin{equation*}
\begin{aligned}
I(\mathscr{H}:B) &= I(\mathscr{H}) - I(\mathscr{H}|B) \\
&\ge \log |\mathcal{H}| - (I(\delta)+\delta \log (|\mathcal{H}|-1)) \\
&> (1-\delta)\log |\mathcal{H}|-I(\delta).
\end{aligned}
\end{equation*}
\end{proof}

\begin{claim}
\label{Claim:entropy}
For $0 < p \le \frac{1}{2}$, $-(1-p)\log (1-p) \le -p \log p$.
\end{claim}

\begin{proof}
Let $f(p) = p \ln p - (1-p) \ln (1-p)$, then $f'(p) = \ln p+\ln (1-p)+2$, so $f''(p) = \frac{2p-1}{p(p-1)}$.
When $0<p \le \frac{1}{2}$, $f''(p) \ge 0$ and thus $f'(p)$ is an increasing function.
Because $\lim\limits_{x\to 0}f'(x) <0$ and $f'(\frac{1}{2}) > 0$, there exists a point $p_0$ $(0 < p_0 < \frac{1}{2})$ such that $f'(p_0) = 0$. Thus, $f(p)$ is a decreasing function when $0<p \le p_0$ and $f(p)$ is an increasing function when $p_0 \le p \le \frac{1}{2}$.
Since $\lim\limits_{x\to 0}f(x) = 0$ and $f(\frac{1}{2}) = 0$, we have $f(p) \le 0$ for any $0<p\le \frac{1}{2}$, i.e, $(1-p) \ln (1-p) \ge p \ln p $. Thus, $-(1-p) \log (1-p) \le -p \log p $.
\end{proof}

\subsection{Upper bound}
\label{sec:upper}
In this section, we give an upper bound on the sample complexity of $\mathsf{HSP}$ by proving Theorem \ref{Theorem:hiddenupper}. Specifically, we propose Algorithm \ref{Sample algorithm1} to solve $\mathsf{HSP}$. Then we analyze the correctness and sample complexity of Algorithm \ref{Sample algorithm1} in Section \ref{sec:cor} and \ref{sec:com},  respectively. The number of examples used in Algorithm \ref{Sample algorithm1} is an upper bound on the sample complexity of $\mathsf{HSP}$.
\begin{theorem}[Upper bound]
\label{Theorem:hiddenupper}
The number of uniform examples sufficient to solve $\mathsf{HSP}$ with bounded error is $O\left(\max \left\{\mathsf{sr}(\mathcal{H}),\sqrt{\max\limits_{H \in \mathcal{H}}\frac{|G|}{|H|}\mathsf{sr}(\mathcal{H})}\right\}\right)$.
\end{theorem}
Algorithm \ref{Sample algorithm1} is shown as follows. In Algorithm \ref{Sample algorithm1}, if $\max\limits_{H\in \mathcal{H}}\frac{|G|}{|H|} > \mathsf{sr}(\mathcal{H})$, let $A = \left\lceil 9\sqrt{\max\limits_{H\in \mathcal{H}}\frac{|G|}{|H|}\mathsf{sr}(\mathcal{H})} \right\rceil$ and $B = \left\lceil\sqrt{\max\limits_{H\in \mathcal{H}}\frac{|G|}{|H|} / \mathsf{sr}(\mathcal{H})}\right\rceil$; if $\max\limits_{H\in \mathcal{H}}\frac{|G|}{|H|} \le \mathsf{sr}(\mathcal{H})$, let $A =  9\max\limits_{H\in \mathcal{H}}\frac{|G|}{|H|}$ and $B = 1$. In this way, we always have $AB \ge 9\max\limits_{H \in \mathcal{H}} \frac{|G|}{|H|}$. If two examples $(a,f(a))$ and $(b,f(b))$ satisfies $f(a) = f(b)$, then we say these two examples collide and call $(a,b)$ a collision pair. Furthermore, if there exists at least a collision pair between two example sets $P$ and $Q$, then we also say that $P$ and $Q$ collide.

\begin{algorithm}
\caption{The sample algorithm of $\mathsf{HSP}$}
\label{Sample algorithm1}
\begin{enumerate}
\item Let $W=\emptyset$.
\item Sample $A$ times to obtain an example set $P$.
\item For $1 \le i \le 9 \cdot\mathsf{sr}(\mathcal{H})$, sample $B$ times to obtain an example set $Q_i$. If $P$ and $Q_i$ collide, then we randomly select a collision pair $(a_i,b_i)$ such that $(a_i,f(a_i))\in P$, $(b_i,f(b_i)) \in Q_i$, and insert $a_i^{-1}b_i$ into set $W$. 
\item Repeat Step 2-3 $\lceil\frac{\ln \frac{1}{\delta}}{\ln \frac{6}{5}}\rceil$ times, return $\langle W \rangle$.
\end{enumerate}
\end{algorithm}

\subsubsection{Correctness Analysis}
\label{sec:cor}
In this section, Our goal is to prove that the probability of Algorithm \ref{Sample algorithm1} failing is no more than $\delta$. By the definition of $\mathsf{HSP}$, we have $f(a_i) = f(b_i)$ if and only if $a_i^{-1}b_i \in H$ for any $i$. Thus, any element added into $W$ in Step 3 is an element in $H$. In the following, it suffices to prove that the probability that $W$ is a generating set of  $H$ is not less than $1-\delta$.

Let $N = \frac{|G|}{|H|}$. We call each execution of Step 2-3 an iteration. In Step 2, let $\zeta_{i,j,l}$ be the indicator random variable for the event that $P_l$ and $Q_{i,j}$ collide, where $P_l$ is the $l$-th sample in $P$ and $Q_{i,j}$ is the $j$-th sample in $Q_i$. Let $\zeta_{i} = \sum_{j,l} \zeta_{i,j,l}$. Then $E(\zeta_{i,j,l}) = Pr\{\zeta_{i,j,l} = 1\}=\frac{1}{N}$ and $D(\zeta_{i,j,l}) =\frac{1}{N}(1-\frac{1}{N})$ for any $i,j,l$. Thus, $E(\zeta_{i}) = \frac{AB}{N}$. Since $\zeta_{i,j_1,l_1}$ and $\zeta_{i,j_2,l_2}$ are independent for any $(j_1,l_1) \neq (j_2,l_2)$, by Chebyshev's Inequality, we have 
\begin{equation*}
\begin{aligned}
Pr\left\{|\frac{\zeta_{i}}{AB}-\frac{1}{N}| \ge \frac{2}{3}\frac{1}{N}\right\}  
&\le \frac{D(\frac{1}{AB}\zeta_{i})}{(\frac{2}{3}\frac{1}{N})^2} & 
&\le\frac{D(\zeta_i)}{(9N)^2(\frac{2}{3}\frac{1}{N})^2} \\
&=\frac{\sum_{j,l}D(\zeta_{i,j,l})}{(9N)^2(\frac{2}{3}\frac{1}{N})^2} &
&=\frac{D(\zeta_{1,1,1})}{9N(\frac{2}{3}\frac{1}{N})^2} \\
&=\frac{\frac{1}{N}(1-\frac{1}{N})}{9N(\frac{2}{3}\frac{1}{N})^2} & &< 1/4.
\end{aligned}
\end{equation*}
Since $AB \ge 9N$,
\begin{equation*}
\begin{aligned}
Pr\{\zeta_i \ge 3\} &\ge Pr\{\zeta_i \ge \frac{AB}{3N}\} \\
&= Pr\{\frac{\zeta_i}{AB} \ge \frac{1}{3}\frac{1}{N}\} \\
&> 1-\frac{1}{4} \\
&= \frac{3}{4},
\end{aligned}
\end{equation*}
so the probability that $P$ and $Q_{i}$ collide is large than $\frac{3}{4}$. 

Let $\beta_i$ be the indicator random variable for the event that $P$ and $Q_i$ collide. Then $E(\beta_i) > \frac{3}{4}$ and $D(\beta_i)=E(\beta_i)(1-E(\beta_i))$ for any $i$. Let $\beta = \sum_{i=1}^{9\mathsf{sr}(\mathcal{H})} \beta_i$. Since $\beta_{i}$ and $\beta_{j}$ are independent for any $i \neq j$, by Chebyshev's Inequality, we have
\begin{equation*}
\begin{aligned}
Pr\{|\frac{\beta}{9\mathsf{sr}(\mathcal{H})}- E(\beta_1)| \ge \frac{2}{3}E(\beta_1)\} &\le \frac{D(\frac{1}{9\mathsf{sr}(\mathcal{H})}\beta)}{(\frac{2}{3}E(\beta_1))^2} &
&= \frac{D(\beta)}{(9\mathsf{sr}(\mathcal{H}))^2(\frac{2}{3}E(\beta_{1}))^2} \\
&= \frac{\sum_{i=1}^{9\mathsf{sr}(\mathcal{H})}D(\beta_i)}{(9\mathsf{sr}(\mathcal{H}))^2(\frac{2}{3}E(\beta_{1}))^2} &
&=\frac{D(\beta_1)}{(9\mathsf{sr}(\mathcal{H}))(\frac{2}{3}E(\beta_{1}))^2} \\
&=\frac{(1-E(\beta_{1}))}{(9\mathsf{sr}(\mathcal{H}))(\frac{4}{9}E(\beta_{1}))} & &< \frac{1}{3\mathsf{sr}(\mathcal{H})}. 
\end{aligned}
\end{equation*}
Thus, 
\begin{equation*}
\begin{aligned}
Pr\{\beta \ge \mathsf{r}(H)\} &\ge Pr\{\beta \ge \mathsf{sr}(\mathcal{H})\} \\
&\ge Pr\{\beta \ge \frac{9 \mathsf{sr}(\mathcal{H})}{4}\} \\
&= Pr\{\beta \ge \frac{9 \mathsf{sr}(\mathcal{H})\cdot E\beta_{1}}{3}\} \\
&= Pr\{\frac{\beta}{9 \mathsf{sr}(\mathcal{H})} \ge \frac{1}{3}E(\beta_1)\} \\
&> 1-\frac{1}{3\mathsf{sr}(\mathcal{H})},
\end{aligned}
\end{equation*}
which means the probability that Algorithm \ref{Sample algorithm1} finds at least $\mathsf{r}(H)$ elements in $H$ is at least $1-\frac{1}{3\mathsf{sr}(\mathcal{H})}$ in each iteration.

Let $E$ be the event that $\mathsf{r}(H)$ independent elements in $H$ construct a generating set of $H$. Let $H_i$ be a subgroup of $H$ generated by the first $i$ elements ($1 \le i \le \mathsf{r}(H)$) and $H_0$ be the trivial subgroup. If $E$ happens, then $H_{i-1} < H_i$ for any $i$. By Lagrange's Theorem, we have $\frac{|H_{i}|}{|H_{i-1}|} \ge 2$. Thus, $\frac{|H|}{|H_i|}  \ge 2^{\mathsf{r}(H)-i}$. In this way, the probability that $E$ happens is 
\begin{equation*}
(1-\frac{|H_0|}{|H|})(1-\frac{|H_1|}{|H|})\cdots (1-\frac{|H_{r(H)-1}|}{|H|}) \ge 
(1-\frac{1}{2^{\mathsf{r}(H)}})(1-\frac{1}{2^{\mathsf{r}(H)-1}})\cdots (1-\frac{1}{2}) > \frac{1}{4},
\end{equation*}
where the last inequality comes from \cite{Watrous2006Simon}. 
Therefore, the probability that the elements added into $W$ in each iterator construct a generating set of $H$ is at least $(1-\frac{1}{3\mathsf{sr}(\mathcal{H})})\frac{1}{4} \ge \frac{2}{3}\cdot\frac{1}{4}=\frac{1}{6}$. As a result, after repeating $\lceil\frac{\ln \frac{1}{\delta}}{\ln \frac{6}{5}}\rceil \ge \log_{\frac{5}{6}} \delta$ times, the probability that $W$ is not a generating set of $H$ is no more than $(1-\frac{1}{6})^{\log_{\frac{5}{6}} \delta} = \delta$, i.e., Algorithm \ref{Sample algorithm1} succeeds with probability at least $1-\delta$.

\subsubsection{Complexity Analysis}
\label{sec:com}
The sample complexity of Algorithm \ref{Sample algorithm1} is $(A+9B \mathsf{sr}(\mathcal{H}))\left\lceil\frac{\ln \frac{1}{\delta}}{\ln \frac{6}{5}}\right\rceil$.
If $\max\limits_{H \in \mathcal{H}}\frac{|G|}{|H|} > \mathsf{sr}(\mathcal{H})$, then $A = \left\lceil 9\sqrt{\max\limits_{H \in \mathcal{H}}\frac{|G|}{|H|} \mathsf{sr}(\mathcal{H})}\right\rceil, B = \left\lceil\sqrt{\max\limits_{H \in \mathcal{H}}\frac{|G|}{|H|} / \mathsf{sr}(\mathcal{H})}\right\rceil$, and thus the sample complexity is $O\left(\sqrt{\max\limits_{H \in \mathcal{H}}\frac{|G|}{|H|} \mathsf{sr}(\mathcal{H})}\right)$;
if $\max\limits_{H \in \mathcal{H}}\frac{|G|}{|H|} \le \mathsf{sr}(\mathcal{H})$, then $A = \max\limits_{H \in \mathcal{H}}\frac{|G|}{|H|}, B = 1$, and thus the sample complexity is $O\left(\mathsf{sr}(\mathcal{H})\right)$. Since $\mathsf{sr}(\mathcal{H}) \ge \sqrt{\max\limits_{H \in \mathcal{H}}\frac{|G|}{|H|}\mathsf{sr}(\mathcal{H})}$ is equivalent to $\mathsf{sr}(\mathcal{H}) \ge \max\limits_{H \in \mathcal{H}}\frac{|G|}{|H|}$, the sample complexity of Algorithm \ref{Sample algorithm1} can be expressed as $O\left(\max \left\{\mathsf{sr}(\mathcal{H}),\sqrt{\max\limits_{H \in \mathcal{H}}\frac{|G|}{|H|}\mathsf{sr}(\mathcal{H})}\right\}\right)$ equivalently, i.e., Theorem \ref{Theorem:hiddenupper} is proved.

\section{Application}
\label{sec:Application}
The  results in Section \ref{sec:Bounds} can be applied to some more specific classes of $\mathsf{HSP}$, including $\mathsf{rAHSP}$ and $\mathsf{GSP}$ defined in Section \ref{sec:resut}. We show the sample complexity of these problems by proving Corollary \ref{corollary:ahsp} and \ref{corollary:gsp}.


\subsection{Abelian hidden subgroup problem}
\begin{proof}[Proof of Corollary \ref{corollary:ahsp}]
In $\mathsf{rAHSP}$, $G = \mathbb{Z}_{p_1}^{n_1} \times \mathbb{Z}_{p_2}^{n_2} \times \cdots \times \mathbb{Z}_{p_m}^{n_m}$, and thus 
$|G| = \prod_{i=1}^m p_i^{n_i}$. Since $H = H_1 \times H_2 \times \cdots \times H_m$, where $H_i \le \mathbb{Z}_{p_i}^{n_i}$ and $\mathsf{r}(H_i) = k_i$ for any $i \in [m]$, we have $|H| = \prod_{i=1}^m p_i^{k_i}$. Hence,
\begin{equation}
\label{eq1}
\frac{|G|}{|H|} = \prod_{i=1}^m p_i^{n_i-k_i},
\end{equation}
so
\begin{equation*}
\log \frac{|G|}{|H|} = \sum_{i=1}^m (n_i-k_i)\log p_i.
\end{equation*}
By a counting method \cite{Stanley2000Enumerative}, the number of subgroup of rank $k$ in $\mathbb{Z}_{p}^{n}$ is 
\begin{equation*}
\label{eq:counting}
\prod_{j=0}^{k-1} \frac{p^{n}-p^j}{p^{k}-p^j} > p^{(n-k)k}.
\end{equation*}
As a result, 
\begin{equation*}
|\mathcal{H}| = \prod_{i=1}^m \prod_{j=0}^{k_i-1} \frac{p_i^{n_i}-p_i^j}{p_i^{k_i}-p_i^j} > \prod_{i=1}^m p_i^{(n_i-k_i)k_i}.
\end{equation*}
Thus
\begin{equation*}
\log |\mathcal{H}| > \sum_{i=1}^m (n_i-k_i)k_i \log p_i.
\end{equation*}
Therefore, 
\begin{equation}
\label{eq2}
\frac{\log |\mathcal{H}|}{\log \frac{|G|}{|H|}} > \frac{\sum_{i=1}^m (n_i-k_i)k_i \log p_i}{\sum_{i=1}^m (n_i-k_i)\log p_i} \ge \min\limits_{i\in [m]} k_i. 
\end{equation}
Moreover, by Claim \ref{Claim:count}, $\mathsf{r}(H) \le \max_{i \in [m]} k_i$ for any $H \in \mathcal{H}$, so 
\begin{equation}
\label{eq3}
\mathsf{sr}(\mathcal{H}) \le \max_{i \in [m]} k_i.
\end{equation}
By substituting \cref{eq1,eq2,eq3} into Theorem \ref{Theorem:main}, we obtain that the number of uniform examples required for learning the hidden subgroup with bounded error is at least $\Omega\left(\max\left\{\min\limits_{i\in [m]} k_i, \min\limits_{i\in [m]} \sqrt{k_i \prod_{j=1}^m p_j^{n_j-k_j}}\right\}\right)$ and at most $O\left(\max\limits_{i\in [m]}\left\{k_i, \sqrt{k_i \prod_{j=1}^m p_j^{n_j-k_j}}\right\} \right)$ for $\mathsf{rAHSP}$.
\end{proof}
\begin{claim}
\label{Claim:count}
For a finite group $G = G_1 \times G_2 \times \cdots \times G_m$, $\mathsf{r}(G) \le \max_{i \in [m]} \mathsf{r}(G_i)$.
\end{claim}
\begin{proof}
Suppose $\mathsf{r}(G_i) = r_i$ and $T_i = \{T_{i1},...,T_{ir_i}\}$ is a generating set of $G_i$ for $i \in [m]$. In the following, we try to construct a generating set of $G$.
For $1 \le i \le m$, let 
\begin{equation*}
    T'_{ij} =
    \begin{cases}
    T'_{ij},& j \le r_i \\
    e_i,& r_i < j \le \max_{i \in [m]}r_i
    \end{cases},
\end{equation*}
where $e_i$ is the identity element of $G_i$.
For $1 \le j \le \max_{i\in [m]}r_i$, let $s_j = (T'_{1j},...,T'_{mj})$ and $S = \{s_1,...,s_{\max_{i \in [m]}r_i}\}$. Let $T'_i$ denote the set of the $i$-th component of the elements in $S$, i.e., $T'_i = \{T'_{i1},...,T'_{i\max_{i \in [m]}r_i}\}$.
Since $T'_i = T_i \cup \{e_i\}$, $T'_i$ is also a generating set of $G_i$. Thus, $S$ is a generating set of $G$, which means $\mathsf{r}(G) \le \max_{i \in [m]} r_i = \max_{i \in [m]} \mathsf{r}(G_i)$.
\end{proof}

\subsection{Generalized Simon's Problem}
\label{sec:gsp}
\begin{proof}[Proof of Corollary \ref{corollary:gsp}]
By substituting $i = 1$, $p_1 = p$, $n_1 = n$, $k_1 = k$ into Corollary \ref{corollary:ahsp}, we find that  the sample complexity of $\mathsf{GSP}$ is at least $\Omega\left(\max\left\{k, \sqrt{k \cdot p^{n-k}}\right\}\right)$ and at most $O\left(\max\left\{k, \sqrt{k \cdot p^{n-k}}\right\}\right)$, i.e, the sample complexity of $\mathsf{GSP}$ is $\Theta\left(\max\left\{k, \sqrt{k \cdot p^{n-k}}\right\}\right)$.
\end{proof}



\section{Conclusion}
\label{sec:Conclusion}
In this paper, we have discussed the classical sample complexity of the hidden subgroup problem ($\mathsf{HSP}$) over finite groups. We have shown the classical sample complexity of $\mathsf{HSP}$ is at least $\Omega\left(\max \left\{\min\limits_{H \in\mathcal{H}}\frac{\log |\mathcal{H}|}{\log \frac{|G|}{|H|}}, \min\limits_{H \in\mathcal{H}}\sqrt{\frac{|G|}{|H|}\frac{\log |\mathcal{H}|}{\log \frac{|G|}{|H|}}}\right\}\right)$ and at most $O\left(\max\limits_{H \in \mathcal{H}} \left\{sr(\mathcal{H}),\sqrt{\frac{|G|}{|H|}sr(\mathcal{H})}\right\}\right)$. Our result may be helpful to clarify the gap between quantum computing and classical computing on this problem. Furthermore, we have applied the result to obtain the sample complexity of some concrete instances of hidden subgroup problem. Particularly, we have obtained a tight bound $\Theta\left(\max\left\{k, \sqrt{k \cdot p^{n-k}}\right\}\right)$ for the sample complexity of $\mathsf{GSP}$. In the future, we will generalize our results to more instances of the hidden subgroup problem, especially for the non-Abelian case. We also believe the information-theoretic approach to obtain the lower bound in this paper will have further application in other learning problems.

\bibliography{GSP}

\end{document}